\documentclass[11pt]{article}
\usepackage[paperwidth=6.9in, paperheight=10in, textwidth=5.25in, textheight=8 in]{geometry}
\usepackage{amsmath,amsthm,amssymb,thmtools}
\usepackage{hyperref}
\usepackage{bookmark}
\usepackage[ruled,vlined]{algorithm2e}
\usepackage{multirow}
\usepackage{tcolorbox}
\usepackage{caption}
\usepackage{colortbl}
\usepackage{boldline, bigdelim, booktabs, cellspace, array, makecell}
\usepackage{tikz}
\usepackage{tabu}
\usepackage{threeparttable}
\usepackage[classfont=caps,funcfont=caps,langfont=caps]{complexity}
\tcbuselibrary{theorems}

\newcommand{\ints}{\ensuremath{\mathbb{Z}}}

\newcommand{\RR}{\ensuremath{\mathbb{R}}}
\newcommand{\posreals}{\ensuremath{\mathbb{R}_{\geq 0}}}

\newcommand{\prob}{\mathrm{Pr}}

\newcommand*{\esp}[1]{\operatorname{\mathbb{E}}[ #1 ]}

\newcommand{\SLS}{\textsc{StreamingLocalSearch}}
\newcommand{\mSLS}{\textsc{MultipassLocalSearch}}
\newcommand{\rSLS}{\textsc{RandomizedLocalSearch}}
\newcommand{\mrSLS}{\textsc{MultipassRandomizedLocalSearch}}

\newcommand{\opt}{\mathit{OPT}}
\newcommand{\fopt}{f\lb \opt \rb}

\newcommand{\Sinit}{S_{\mathrm{init}}}

\newcommand{\gammaoff}{\bar{\gamma}_{\mathrm{off}}}

\newcommand{\e}{\varepsilon}
\newcommand{\bb}{\backslash}
\makeatletter
\newcommand{\thickhline}{%
    \noalign {\ifnum 0=`}\fi \hrule height 1pt
    \futurelet \reserved@a \@xhline
}
\newcolumntype{"}{@{\hskip\tabcolsep\vrule width 1pt\hskip\tabcolsep}}
\newcolumntype{?}{!{\vrule width 1pt}}
\makeatother

\newtheorem{theorem}{Theorem}
\newtheorem{lemma}{Lemma}[section]

\newcommand{\lc}{\!\left\{}
\newcommand{\rc}{\right\}}
\newcommand{\ld}{\!\left[}
\newcommand{\rd}{\right]}
\newcommand{\lb}{\!\left(}
\newcommand{\rb}{\right)}
\newcommand{\card}[1]{\left\vert #1\right\vert}
\DeclareMathOperator*{\argmin}{arg\,min}

\DeclareMathOperator*{\argmax}{arg\,max}

\newcommand{\bigO}[1]{O\lb #1 \rb}

\newcommand{\cI}{\mathcal{I}}
\newcommand{\cM}{\mathcal{M}}

\SetArgSty{textnormal}

\begin{document}

\title{Improved Multi-Pass Streaming Algorithms for Submodular Maximization with Matroid Constraints\thanks{A preliminary version of this work was presented at the International Conference on Approximation Algorithms for Combinatorial Optimization Problems (APPROX 2020). This work was funded by the grants ANR-19-CE48-0016 and ANR-18-CE40-0025-01 from the French National Research Agency (ANR). This work was supported by EPSRC New Investigator Award EP/T006781/1.}
}


\author{Chien-Chung Huang\footnote{CNRS, DI ENS, Universit\'{e} PSL, Paris,  France, \href{chien-chung.huang@ens.fr}{chien-chung.huang@ens.fr}}        \and
        Theophile Thiery$^{\ddagger}$ \and
        Justin Ward\footnote{School of Mathematical Sciences, Queen Mary University of London, United Kingdom, Emails: \{\href{t.f.thiery}{t.f.thiery}, \href{justin.ward}{justin.ward}\}@qmul.ac.uk}
}


\date{}

\maketitle
\vspace{-2em}
\begin{abstract}
    We give improved multi-pass streaming algorithms for the problem of maximizing a monotone or arbitrary non-negative submodular function subject to a general $p$-matchoid constraint in the model in which elements of the ground set arrive one at a time in a stream. The family of constraints we consider generalizes both the intersection of $p$ arbitrary matroid constraints and $p$-uniform hypergraph matching. For monotone submodular functions, our algorithm attains a guarantee of $p+1+\e$ using $O(p/\e)$-passes and requires storing only $O(k)$ elements, where $k$ is the maximum size of feasible solution. This immediately gives an $O(1/\e)$-pass $(2+\e)$-approximation algorithms for monotone submodular maximization in a matroid and $(3+\e)$-approximation for monotone submodular matching. Our algorithm is oblivious to the choice $\e$ and can be stopped after any number of passes, delivering the appropriate guarantee. We extend our techniques to obtain the first multi-pass streaming algorithm for general, non-negative submodular functions subject to a $p$-matchoid constraint with a number of passes independent of the size of the ground set and $k$. We show that a randomized $O(p/\e)$-pass algorithm storing $O(p^3k\log(k)/\e^3)$ elements gives a $(p+1+\gammaoff+O(\e))$-approximation, where $\gammaoff$ is the guarantee of the best-known offline algorithm for the same problem.
    \end{abstract}

    \section{Introduction}
    \label{sec:introduction}
    Many discrete optimization problems in theoretical computer science, operations research, and machine learning can be cast as special cases of maximizing a \emph{submodular} function $f$ subject to some constraint. Formally, a function $f : 2^X \to \posreals$ is submodular if and only if $f(A) + f(B) \geq f(A \cup B) + f(A \cap B)$ for all $A,B \subseteq X$. One reason for the ubiquity of submodularity in optimization settings is that it also captures a natural ``diminishing returns'' property. Let $f(e \mid A) \triangleq f(A + e) - f(A)$ be the \emph{marginal increase} obtained in $f$ when adding an element $e$ to a set $A$ (where here and throughout we use the shorthands $A+e$ and $A-e$ for $A \cup \{e\}$ and $A \bb \{e\}$, respectively). It is well-known that $f$ is submodular if and only if $f(e \mid B) \leq f(e \mid A)$ for any $A \subseteq B$ and any $e \not\in B$. If additionally we have $f(e \mid A) \geq 0$ for all $A$ and $e \not\in A$ we say that $f$ is \emph{monotone}.

    Here, we consider the problem of maximizing both monotone and arbitrary submodular functions subject to an arbitrary \emph{$p$-matchoid constraint} on the set of elements that can be selected. Formally, a $p$-matchoid $\cM^p = (\cI^p,X)$ on $X$ is given by a collection of matroids $\{\cM_i = (X_i,\cI_i)\}$ each defined on some subset of $X$, where each $e \in X$ is present in at most $p$ of these subsets. A set $S \subseteq X$ is then independent if and only if $S \cap X_i \in \cI_i$ for each matroid $\cM_i$. One can intuitively think of a $p$-matchoid as a collection of matroids in which each element ``participates'' in at most $p$ of the matroid constraints. The resulting family of constraints is quite general and captures both intersections of $p$ matroid constraints (by letting $X_i = X$ for all $\cM_i$) and matchings in $p$-uniform hypergraphs (by considering $X$ as a collection of hyperedges and defining a uniform matroid constraint for each vertex, ensuring that at most one hyperedge containing this vertex is selected).

    In many applications of submodular optimization, such as summarization~\cite{Badanidiyuru:2014ib,lin:2010wpa, mirzasoleiman2016fast,mirzasoleiman18streaming} we must process datasets so large that they cannot be stored in memory. Thus, there has been recent interest in \emph{streaming} algorithms for submodular optimization problems. In this context, we suppose the ground set $X$ is initially unknown and elements arrive one-by-one in a stream. We suppose that the algorithm has an efficient oracle for evaluating the submodular function $f$ on any given subset of $X$, but has only enough memory to store a small number of elements from the stream. Variants of standard greedy and local search algorithms have been developed that obtain a constant-factor approximation in this setting, but their approximation guarantees are considerably worse than that of their simple, offline counterparts.

    Here, we consider the \emph{multi-pass} setting in which the algorithm is allowed to perform several passes over a stream---in each pass all of $X$ arrives in some order, and the algorithm is still only allowed to store a small number of elements. In the \emph{offline} setting, simple variants of greedy~\cite{FisherNemhauserWolsey} or local search~\cite{Feldman2011,Lee:2010} algorithms in fact give the best-known approximation guarantees for maximizing submodular functions subject to the $p$-matroid constraints or a general $p$-matchoid constraint.
    However, these algorithms potentially require considering all elements in $X$ \emph{each time a choice is made}. It is natural to ask whether this is truly necessary, or whether we could instead recover an approximation ratio nearly equal to these offline algorithms by using only a constant number of passes through the data stream.

    \subsection{Our Results}
    \label{sec:our-results}
    Here we show that for monotone submodular functions, $O(1/\e)$-passes suffice to obtain guarantees only $(1+\e)$ times worse than those guaranteed by the offline local search algorithm. We give an $O(p/\e)$-pass streaming algorithm that gives a $p+1+\e$ approximation for maximizing a monotone submodular function subject to an arbitrary $p$-matchoid constraint. It immediately gives us an $O(1/\e)$-pass streaming algorithm attaining a $2+\e$ approximation for matroid constraints and a $3+\e$ approximation for matching constraints in graphs.
    Each pass of our algorithm is equivalent to a single pass of the streaming local search algorithm described by Chakrabarti and Kale~\cite{DBLP:journals/mp/ChakrabartiK15} and Chekuri, Gupta, and Quanrud~\cite{DBLP:conf/icalp/ChekuriGQ15}. However, obtaining a rapid convergence to a $p+1+\e$ approximation requires some new insights. We show that if a pass makes either large or small progress in the value of $f$, then the guarantee obtained at the end of this pass can be improved. Balancing these two effects then leads to a carefully chosen sequence of parameters for each pass. Our general approach is similar to that of Chakrabarti and Kale~\cite{DBLP:journals/mp/ChakrabartiK15}, but our algorithm is \emph{oblivious} to the choice of $\e$.  This allows us to give a uniform bound on the convergence of the approximation factor obtained after some number $d$ of passes. This bound is actually available to the algorithm, and so we can certify the quality of the current solution after each pass. In practice, this allows for terminating the algorithm early if a sufficient guarantee has already been obtained. Even in the worst case, however, we improve on the number of passes required by similar previous results by a factor of  $O(\e^{-2})$. Our algorithm only requires storing $O(k)$ elements, where $k$ is the \emph{rank} of the given $p$-matchoid, defined as the size of the largest independent set of elements.

    Building on these ideas, we also give a randomized, multi-pass algorithm that uses $O(p/\e)$-passes and attains a $p+1+\gammaoff+O(\e)$ approximation for maximizing an arbitrary submodular function subject to a $p$-matchoid constraint, where $\gammaoff$ is the approximation ratio attained by best-known offline algorithm for the same problem. To the best of our knowledge, ours is the first multipass algorithm when the function is non-monotone with a number of passes independent of $n$ and $k$, where $n$ is the size of the ground set. In this case, our algorithm requires storing $O(p^3 k\log k/\e^3)$ elements. We remark that to facilitate comparison with existing work, we have stated all approximation guarantees as factors $\gamma \geq 1$. However, we note that if one states ratios of the form $1/\gamma$ less than 1, then our results lead to $1/\gamma - \e$ approximations in which all dependence on $p$ can be eliminated (by setting simply selecting some $\e' = p\e$).

    \subsection{Related Work}
    \begin{tcolorbox}[
      colframe=gray!75!black,
      colback=gray!7!white,
      size=small,
      title= Current State of the Art]
      \begin{center}
        \small{
          \begin{tabular}{|>{\columncolor{white}}c !{\vrule width 1pt}>{\columncolor{white}}c|>{\columncolor{white}}c !{\vrule width 1pt} >{\columncolor{white}}c|>{\columncolor{white}}c?}
            \hline
            &  \multicolumn{2}{c?}{\cellcolor{white}Offline} & \multicolumn{2}{c?}{\cellcolor{white} Streaming} \\ \cmidrule{2-5} \cellcolor{white}\multirow{-2}{*}{\cellcolor{white}Constraint}& M & NN & M & NN \\
            \thickhline
            \small matroid & \small$e/(e-1)$ \cite{Calinescu:2011ju} & \small$ 2.598$ \cite{Buchbinder:2016} & \small$4$~\cite{DBLP:journals/mp/ChakrabartiK15,DBLP:conf/icalp/ChekuriGQ15,Feldman2018DoLess}\!  & \small$5.8284$ \cite{Feldman2018DoLess} \\
            \hline
            \small $(p,b)$-hyp.m\! & \small$p+\e$ \cite{Feldman2011} & \small$\frac{p^2 + \e}{p-1}$ \cite{Feldman2011} \!& \small$4p$ \cite{DBLP:conf/icalp/ChekuriGQ15,Feldman2018DoLess} & \small$4p + 2 -o(1)$ \cite{Feldman2018DoLess}\! \\
            \hline
            \small $p$-mat.int & \small$p+ \e$ \cite{Lee:2010} & \small$\frac{p^2 + (p-1)\e}{p-1}$ \cite{Lee:2010}\!& \small$4p$ \cite{DBLP:journals/mp/ChakrabartiK15,DBLP:conf/icalp/ChekuriGQ15,Feldman2018DoLess}\!  & \small$4p + 2 -o(1)$ \cite{Feldman2018DoLess}\! \\
            \hline
            & & $\tfrac{ep}{(1-\e)(2-o(1))}$&  & \\
            \cellcolor{white}\multirow{-2}{*}{$p$-matchoid}& \cellcolor{white}\multirow{-2}{*}{$p+1$ \cite{Badanidiyuru:2013jc,FisherNemhauserWolsey}}& \cite{DBLP:journals/siamcomp/ChekuriVZ14,DBLP:conf/focs/FeldmanNS11}& \cellcolor{white}{\multirow{-2}{*}{$4p$ \cite{DBLP:conf/icalp/ChekuriGQ15,Feldman2018DoLess}}}& \cellcolor{white}{\multirow{-2}{*}{$4p + 2 -o(1)$ \cite{Feldman2018DoLess}\!}}\\
            \hline
           \end{tabular}
      \captionof{table}{Approximation ratio in offline and streaming setting}
       \label{tab:1}
      }
      \vspace{0.2cm}
       \small{
      \begin{tabular}{|>{\columncolor{white}}c?>{\columncolor{white}}c|>{\columncolor{white}}c|>{\columncolor{white}}c?>{\columncolor{white}}c|>{\columncolor{white}}c|>{\columncolor{white}}c?}
        \hline
       \cellcolor{white} & \multicolumn{3}{c?}{\cellcolor{white}Multipass} & \multicolumn{3}{c?}{\cellcolor{white}Our results}\\
        \cmidrule{2-7}
        \multirow{-2}{*}{\cellcolor{white}Constraint}& M & $\#$-passes & \!\!NN\!\!& M& NN & \!\!$\#$-passes \!\!\\
        \thickhline
        \small matroid & \small$2+\e$ \cite{DBLP:journals/mp/ChakrabartiK15}  & \small$O(1/\e^{3})$ \cite{DBLP:journals/mp/ChakrabartiK15}\! &$\ast$ & $2 + \e$ & $4.589 + \e$ &\!\!$O(1/\e)$\!\!\\
        \hline
        &  \!\!\! \!\!\!\!\!\!& &  & & & \\
        \multirow{-1}{*}{\cellcolor{white}$(p,b)$-hyp.m\!} & \multirow{-1}{*}{\!\!\!\small $p+1+\e$\!\!\! }\!& \multirow{-1}{*}{\!\!\!\small$O(p^4\log(p)/\e^3)$\!\!\!} & \multirow{-1}{*}{\cellcolor{white} \!\!$\ast$} & \multirow{-1}{*}{\cellcolor{white} \!\!\!\!$p+1+\e$\!\!\!}& &\multirow{-1}{*}{\cellcolor{white}\!$O(p/\e)$} \!\!\\
        & \!\!\! \cite{DBLP:journals/mp/ChakrabartiK15} \!\!\!\!\!\!& \!\!\! \cite{DBLP:journals/mp/ChakrabartiK15}\! \!\!\! &  & &\multirow{-3}{*}{\cellcolor{white}\vspace{-0.3em}\makecell{$p+1 + O(\e)$\\$ +\tfrac{p^2}{p-1}$}} & \\
        \hline
        & & &  & & & \\
        \multirow{-1}{*}{\cellcolor{white}$p$-mat.int\!}& \multirow{-1}{*}{\cellcolor{white}\!\!\!\small$p+1+\e$ \!\!\!\!\!\!}& \multirow{-1}{*}{\!\!\!\small$O(p^4\log(p)/\e^3)$\!\!\! }\!\!& \multirow{-1}{*}{\cellcolor{white} \!\!$\ast$} & \multirow{-1}{*}{\cellcolor{white} \!\!\!\!$p+1+\e$\!\!\!}& & \multirow{-1}{*}{\cellcolor{white}\!\!\!$O(p/\e)$}\!\!\\
        & \cellcolor{white}\cite{DBLP:journals/mp/ChakrabartiK15}& \cite{DBLP:journals/mp/ChakrabartiK15}\! &  & &\!\!\! \!\!\!\!\!\!\multirow{-3}{*}{\cellcolor{white}\vspace{-0.3em}\makecell{$p+1 + O(\e)$\\$ +\tfrac{p^2}{p-1} $}\!\!\!\!}\!\!\!& \\
        \hline
        \small $p$-matchoid & $\ast$ & $\ast$ & $\ast$ &\!\!\!\!$p+1+\e$\!\!\!\!\! &\makecell{$p+ 1+  O(\e)$\\$+\tfrac{ep}{(1- \e)(2 - o(1))}$} &\!\!\!$O(p/\e)$\!\!\\
        \hline
       \end{tabular}
       \captionof{table}{Summary of results for maximizing a submdodular function in the multipass streaming.}
        \label{tab:2}
       }
      \end{center}
      {We use the following abbreviations:
       M means monotone and NN means that $f$ is non-negative. $p$-mat.int means $p$-matroid intersection and $(p,b)$-hyp.m denotes rank $p$-hypergraph $b$-matching.\\
       $\ast$: If we restrict ourselves with algorithms performing $O(\poly(\e, p))$-passes then only the $1$-pass setting is understood.}
    \end{tcolorbox}

    There is a vast literature on submodular maximization with various constraints and different models of computation.
    In the offline model, the work on maximizing a monotone submodular function goes back to Nemhauser, Wolsey and Fischer \cite{Nemhauser:1978dm}.
    Monotone submodular functions are well studied and many new and powerful results have been obtained since then. The best approximation algorithm under a matroid constraint is due to Calinescu et al. \cite{Calinescu:2011ju} which is the best that can be done using a polynomial number of queries~\cite{Nemhauser:1978dm} (if $f$ is given as a value oracle) or assuming $\P \neq \NP$~\cite{Feige1998} (if $f$ is given explicitly). For more general constraints, Lee, Sviridenko and Vondr\'ak obtained a $p+\e$ approximation algorithm under $p$-matroid intersection constraint \cite{Lee:2010}. Feldman et al.~\cite{Feldman2011} obtained the same approximation ratio for the general class of $p$-exchange systems. For general $p$-matchoid constraints, the best approximation ratio is $p+1$, which is attained by the standard greedy algorithm~\cite{FisherNemhauserWolsey}.

    Non-monotone objectives are less understood even under the simplest assumptions. The current best-known result for maximizing a submodular function under a matroid constraint is $ 2.598$ \cite{Buchbinder:2016}, which is far from the $2.093$ hardness result \cite{Gharan:2011}. Table~\ref{tab:1} gives the best known bounds for the constraints that we consider in the paper.

    Due to the large volume of data in modern applications, there has also been a line of research focused on developing \emph{fast} algorithms for submodular maximization \cite{Badanidiyuru:2013jc, Mirzasoleiman:2015}. However, all results we have discussed so far assume that the entire instance is available at any time, which may not be feasible for massive datasets. This has motivated the study of streaming submodular maximization algorithms with low memory requirements. Badaniyuru et al. \cite{Badanidiyuru:2014ib} achieved a $2 + \e $ approximation algorithm for maximizing a monotone submodular function under a cardinality constraint in the streaming setting. This was recently shown to be the best possible bound attainable in one pass with memory sublinear in the size of the instance~\cite{Feldman2020}. Chakrabarti and Kale~\cite{DBLP:journals/mp/ChakrabartiK15} gave a $4p$ approximation for $p$-matroid intersection constraint or $p$-uniform hypergraph matching. Later, Chekuri et al. \cite{DBLP:conf/icalp/ChekuriGQ15} generalized their argument to arbitrary $p$-matchoid constraints, and also gave a modified algorithm for handling non-monotone submodular objectives. A fast, randomized variant of the algorithm of~\cite{DBLP:journals/mp/ChakrabartiK15} was studied by Feldman, Karbasi and Kazemi \cite{Feldman2018DoLess}, who showed that it has the same approximation guarantee when $f$ is monotone and achieves a $2p + 2\sqrt{p(p+1)} + 1 = 4p + 2 - o(1)$ approximation for general submodular function.
    Related to our work, there is an active research direction focusing on streaming (sub)modular maximization subject to matching constraints. For submodular maximization, the best approximation is $3 + 2 \sqrt{2}$ and $4 + 2\sqrt{3}$ for monotone and non-montone functions respectively \cite{Levin:2020:Streaming}.

    When multiple passes through the stream are allowed, less is known and the tradeoff between the approximation guarantee and the number of passes requires more attention. Assuming cardinality constraints, one can obtain a $\tfrac{e}{e-1} + \e$ multipass streaming algorithm in $\bigO{1/\e}$-passes (see \cite{Badanidiyuru:2013jc,HuangKakimuraMultiPass2018,mcGregor2017,mirzasoleiman2016fast,Ashkan2018}).
    Huang et al. \cite{HuangKakimuraMultiPass2018} achieved a $2+  \e$ approximation under a knapsack constraint in $\bigO{1/\e}$ passes. For the intersection of $p$-partition matroids or rank $p$-hypergraph matching, the number of passes becomes dependent on $p$. Chakrabarti and Kale~\cite{DBLP:journals/mp/ChakrabartiK15}\footnote{In  \cite{DBLP:journals/mp/ChakrabartiK15} a bound of $O(\log p/\e^3)$ is stated. We note that there appears to be a small oversight in their analysis, arising from the fact that their convergence parameter $\kappa$ in this case is $O(\e^3/p^4)$. In any case, it seems reasonable to assume that $p$ is a small constant in most cases.} showed that if one allows  $\bigO{p^4\log(p)/\e^3}$-passes, a $p+1+ \e$ approximation is possible. Here we show how to obtain the same guarantee for an arbitrary $p$-matchoid constraint, while reducing the number of passes to $O(p/\e)$.

    \section{The main multi-pass streaming algorithm}
    \label{sec:multi-pass-main}
    For monotone functions, our main multi-pass algorithm is given by the procedure \mSLS{} in Algorithm~\ref{alg:multi-pass}. We suppose that we are given a submodular function $f : 2^X \to \posreals$ and a $p$-matchoid constraint $\cM^p = (\cI^p,X)$ on $X$ given as a collection of matroids $\{\cM_i = (X_i,\cI_i)\}$. Our procedure runs for $d$ passes, each of which uses a modification of the algorithm of Chekuri, Gupta, and Quanrud~\cite{DBLP:conf/icalp/ChekuriGQ15}, given as the procedure \SLS. In each pass, procedure \SLS{} maintains a current solution $S$, which is initially set to some $\Sinit$.  Whenever an element $x \in \Sinit$ arrives again in the subsequent stream, the procedure simply discards $x$. For all other elements $x$, the procedure invokes a helper procedure \textsc{Exchange}, given formally in Algorithm~\ref{alg:exchange}, to find an appropriate set $C_x \subseteq S$ of up to $p$ elements so that $S \bb C_x + x \in \cI$. It then exchanges $x$ with $C_x$ if it gives a significantly improved solution.
    \begin{algorithm}[t]
    \SetKwProg{myproc}{procedure}{}{}
    \myproc{$\textsc{MultipassLocalSearch}(\alpha,\beta_1,\ldots,\beta_d)$}{
    $S_0 \gets \emptyset$\;
    \For{$i = 1$ \KwTo $d$}{
      Let $\tilde{S}$ be the output of $\textsc{StreamingLocalSearch}(\alpha,\beta_i,S_{i-1})$\;
      $S_{i} \gets \tilde{S}$\;
    }
    \Return $S_{d}$\;
    }
    \BlankLine
    \myproc{$\textsc{StreamingLocalSearch}(\alpha,\beta,\Sinit)$}{
    $S \gets \Sinit$\;
    \ForEach{$x$ in the stream}
    {
    \lIf{$x \in \Sinit$}{discard $x$}
    $C_x \gets \textsc{Exchange}(x,S)$\;
    \If{$f(x | S) \geq \alpha + (1 + \beta)\sum_{c \in C_x}\nu(c,S)$}{
        $S \gets S \bb C_x + x$\;
      }
    }
    \Return $S$\;
    }
    \caption{The multi-pass streaming local search algorithm}
    \label{alg:multi-pass}
    \end{algorithm}
    \begin{algorithm}[t]
    \caption{The procedure $\textsc{Exchange}(x,S)$}
    \label{alg:exchange}
    \SetKwProg{myproc}{procedure}{}{}
    \myproc{$\textsc{Exchange}(x,S)$}{
    $C_x \gets \emptyset$\;
    \ForEach{$\cM_\ell=(X_\ell,\cI_\ell)$ with $x \in X_\ell$}
    {
      $S_\ell \gets S \cap X_\ell$\;
      \If{$S_\ell + x \not\in \cI$}{
        $T_\ell \gets \{y\in S_\ell: S_\ell - y + x \in \cI_\ell\}$\;
        $C_x \gets C_x + \argmin_{t \in T_\ell}\nu(t,S)$\;
      }
    }
    \Return $C_x$\;
    }
    \end{algorithm}
    The improvement is measured with respect to a set of auxiliary weights $\nu(x,S)$ maintained by the algorithm.  For $u,v \in X$, let $u \prec v$ denote that ``element $u$ arrives before $v$'' in the stream.
    Then, we define the \emph{incremental value} of an element $e$ with respect to a set $T$ as
    \[
    \nu(e,T) = f(e \mid \{t' \in T:t' \prec e\})\,.
    \]
    There is a slight difficulty here in that we must also define incremental values for the elements of $\Sinit$. To handle this difficulty, we in fact define $\prec$ with respect to a \emph{pretend} stream ordering. Note that in all invocations of the procedure \SLS{} made by \mSLS, the set $\Sinit$ is either $\emptyset$ or the result of a previous application of \SLS. In our pretend ordering ($\prec$) all of $\Sinit$ first arrives in the same relative pretend ordering as the previous pass, followed by all of $X \bb \Sinit$ in the same order given by the stream $X$. We then define our incremental values with respect to this pretend stream ordering.

    Using these incremental values, \SLS{} proceeds as follows. When an element $x \not\in \Sinit$ arrives, \SLS{} computes a set of elements $C_x \subseteq S$ that can be exchanged for $x$.
    \SLS{} replaces $C_x$ with $x$ if and only if the marginal value $f(x\mid S)$ with respect to $S$ is at least $(1+\beta)$ times larger than the sum of the current incremental values $\nu(c,S)$ of all elements $c \in C_x$ plus some threshold $\alpha$, where $\alpha, \beta > 0$ are given as parameters. In this case, we say that the element $x$ is \emph{accepted}. Otherwise, we say that $x$ is \emph{rejected}.  An element $x \in S$ that has been accepted may later be removed from $S$ if $x \in C_y$ for some later element $y$ that arrives in the stream. In this case we say that $x$ is \emph{evicted}.

    The approximation ratio obtained by one pass of \SLS{} depends on the parameter $\beta$ in two ways, which can be intuitively understood in terms of the standard analysis of the offline local search algorithm for the problem. Intuitively, if $\beta$ is chosen to be too large, more valuable elements will be rejected upon arrival and so, in the offline setting, our solution would be only approximately locally optimal, leading to a deterioration of the guarantee by a factor of $(1+\beta)$. However, in the streaming setting, the algorithm only attempts to exchange an element upon its arrival, and so the final solution will not necessarily be even $(1+\beta)$-approximately locally optimal---an element $x$ may be rejected because $f(x \mid S)$ is small when it arrives, but the processing of later elements in the stream can evict some elements of $S$. After these evictions, we could have $f(x \mid S)$ larger. The key observation in the analyses of~\cite{DBLP:journals/mp/ChakrabartiK15,DBLP:conf/icalp/ChekuriGQ15} is that the total value of these evicted elements---and so also the total increase in the marginal value of all rejected elements---can be bounded by $O(\frac{1}{\beta})$ times the final value of $f(S)$ at the end of the algorithm.
    Intuitively, if $\beta$ is chosen to be too small, the algorithm will make more exchanges, evicting more elements, which may result in rejected elements being much more valuable with respect to the final solution. Selecting the optimal value of $\beta$ thus requires balancing these two effects.

    Here, we observe that this second effect depends only on the total value of those elements that were accepted \emph{after} an element arrives. To use this observation, we measure the ratio $\delta = f(\Sinit)/f(\tilde{S})$ between the value of the initial solution $\Sinit$ of some pass of \SLS{} and the final solution $\tilde{S}$ produced by this pass. If $\delta$ is relatively small---and so one pass makes a lot of progress---then this pass gives us an improvement of $\delta^{-1}$ over the ratio already guaranteed by the previous pass since $f(\tilde{S}) = \delta^{-1}f(\Sinit)$.
    On the other hand, if $\delta$ is relatively large---and so one pass does not make much progress---then the total increase in the value of our rejected elements can be bounded by $\frac{1-\delta}{\beta}f(\tilde{S})$, and so the potential loss due to only testing these elements at arrival is relatively small. Balancing these two effects allows us to set $\beta$ smaller in each subsequent passes and obtain an improved guarantee.

    We now turn to the analysis of our algorithm. Here we focus on a single pass of \SLS{}. For $T,U \subseteq X$ we let $f(T \mid U) \triangleq f(T \cup U) - f(U)$. Throughout, we use $S$ to denote the current solution maintained by this pass (initially, $S = \Sinit$). The following key properties of incremental values will be useful in our analysis. We defer the proof to the Appendix.
    \begin{restatable}{lemma}{incremental}
    \label{lem:incremental-properties}
    For any $T \subseteq U \subseteq X$,
    \begin{enumerate}
    \item $\sum_{e \in T}\nu(e,T) = f(T) - f(\emptyset)$.
    \item $\nu(e,U) \leq \nu(e,T)$ for all $e \in T$.
    \item $f(T \mid U \bb T) \leq \sum_{t \in T}\nu(t,U)$.
    \item At all times during the execution of \SLS{}, $\nu(e,S) \geq \alpha$ for all $e \in S$.
    \end{enumerate}
    \end{restatable}
    Let $A$ denote the set of elements \emph{accepted} during the present pass. These are the elements which were present in the solution $S$ at some previous time during the execution of this pass. Initially we have $A = S = \Sinit$ and whenever an element is added to $S$, during this pass we also add this element to $A$. Let $\tilde{A}$ and $\tilde{S}$ denote the sets of elements $A$ and $S$ at the end of this pass. Note that we regard all elements of $\Sinit$ as having been accepted at the start of the pass. The following lemma follows from the analysis of Chekuri, Gupta, and Quanrud~\cite{DBLP:conf/icalp/ChekuriGQ15} in the single-pass setting. We give a complete, self-contained proof in Appendix~\ref{sec:proof-lemma-refl}. Each element $e \in \tilde{A} \bb \tilde{S}$ was accepted but later \emph{evicted} by the algorithm. For any such evicted element, we let $\chi(e)$ denote the value of $\nu(e,S)$ at the moment that $e$ was removed from $S$.
    \begin{restatable}{lemma}{onepass}
    \label{lem:opt-R}
    Let $f :  2^X \to \posreals$ be a submodular function. Suppose $\tilde{S}$ is the solution produced at the end of one pass of \SLS{} and $\tilde{A}$ be the set of all elements accepted during this pass. Then,
    \[
      f(\opt \cup \tilde{A}) \leq (p + \beta p - \beta)\sum_{e \in \tilde{A}\bb\tilde{S}}\chi(e) + (p + \beta p + 1)f(\tilde{S}) + k\alpha\,.
    \]
    \end{restatable}

    We now derive a bound for the  summation $\sum_{e \in \tilde{A} \bb \tilde{S}}\chi(e)$ (representing the value of \emph{evicted} elements) in terms of the total gain $f(\tilde{S}) - f(\Sinit)$ made by the pass, and also bound the total number of accepted elements in terms of $f(\opt)$.
    \begin{lemma}
    \label{lem:exit-values}
    Let $f : 2^X \to \posreals$ be a submodular function. Suppose that $\tilde{S}$ is the solution produced at the end of one pass of \SLS{} and $\tilde{A}$ is the set of all elements accepted during this pass. Then, $|\tilde{A}| \leq f(\opt)/\alpha$ and
    \[
    \sum_{e \in \tilde{A}\bb \tilde{S}}\chi(e) \leq \frac{1}{\beta}\left(f(\tilde{S}) - f(\Sinit)\right)\,.
    \]
    \end{lemma}
    \begin{proof}
    We consider the quantity $\Phi(A) \triangleq \sum_{e \in A \bb S}\chi(e)$. Suppose some element $a$ with $C_a \neq \emptyset$ is added to $S$ by the algorithm, evicting the elements of $C_a$. Then (as each element can be evicted only once) $\Phi(A)$ increases by precisely $\Delta \triangleq \sum_{e \in C_a}\chi(e)$. Let $S^-_a,S^+_a$ and $A^-_a,A^+_a$ be the sets $S$ and $A$, respectively, immediately before and after $a$ is accepted. Let $\delta_a: = f(S_a^+) - f(S_a^-)$ be the change in the objective function after the exchange between $a$ and $C_a$. Since $a$ is accepted, we must have $f(a \mid S_a^-) \geq \alpha + (1+\beta)\sum_{e \in C_a}\nu(e,S_a^-)$. Then,
    \begin{align*}
    \delta_a
    &= f(S_a^- \bb C_a + a) - f(S_a^-), \\
    &= f(a \mid S_a^- \bb C_a) - f(C_a \mid S_a^- \bb C_a), & \\
    & \geq f(a \mid S_a^-) - f(C_a \mid S_a^- \bb C_a), & \text{(by submodularity)}\\
    & \geq f(a \mid S_a^-) - \sum_{e \in C_a} \nu(e, S_a^-),&  \text{(by Lemma~\ref{lem:incremental-properties}~(3))}\\
    & \geq \alpha + (1+\beta)\sum_{e \in C_a} \nu(e,S_a^-) - \sum_{e \in C_a} \nu(e,S_a^-), & \text{(since $a$ is accepted)}\\
    & = \alpha + \beta\sum_{e \in C_a} \chi(e)\,\text{(by definition of $\chi(e))$} \\
    &= \alpha + \beta\Delta.
    \end{align*}
    It follows that whenever $\Phi(A)$ increases by $\Delta$, $f(S)$ must increase by at least $\beta\Delta$.
    Initially, $\Phi(A) = 0$ and $f(S) = f(\Sinit)$ and at the end of the algorithm, $\Phi(A) = \sum_{e \in \tilde{A} \bb \tilde{S}}\chi(e)$ and $f(S) = f(\tilde{S})$. Thus, $\beta\sum_{e \in \tilde{A}\bb\tilde{S}}\chi(e) \leq [f(\tilde{S}) - f(\Sinit)]$.

    It remains to show that $|\tilde{A}| \leq f(\opt)/\alpha$. For this, we note that the above chain of inequalities also implies that every time an element is accepted (and so $|A|$ increases by one), $f(S)$ also increases by at least $\alpha$. Thus, we have $f(\opt) \geq f(\tilde{S}) \geq \alpha|\tilde{A}|$.
    \end{proof}
    Using Lemma~\ref{lem:exit-values} to bound the sum of exit values in Lemma~\ref{lem:opt-R} then immediately gives us the following guarantee for each pass performed in \mSLS{}. In the $i^{\textrm{th}}$ such pass, we will have $\Sinit = S_{i-1}$, $\tilde{S} = S_{i}$, and $\beta = \beta_i$. We let $A_i$ denote the set of $\tilde{A}$ of all elements accepted during this particular pass.
    \begin{lemma}
    \label{lem:single-pass}
    Let $f : 2^X \to \posreals$ be a submodular function. Consider the $i^\textrm{th}$ pass of \SLS{} performed by \mSLS, and let $A_i$ be the set of all elements accepted during this pass. Then, $|A_i| \leq f(\opt)/\alpha$ and
    \[
    f(\opt \cup A_i) \leq \left(p/\beta_i + p - 1\right)[f(S_i) - f(S_{i-1})] + (p  +  p \beta_i+ 1)f(S_{i}) + k\alpha\,.
    \]
    \end{lemma}


    \section{Analysis of the multipass algorithm for monotone functions.}
    \label{sec:select-param-beta_i}

    We now show how to use Lemma~\ref{lem:single-pass} together with a careful selection of parameters $\alpha$ and $\beta_1,\ldots,\beta_d$ to derive guarantees for the solution $f(S_i)$ produced after the $i^{\mathrm{th}}$ pass made in \mSLS. Here, we consider the case that $f$ is a \emph{monotone} function. In this case, we have $f(\opt) \geq f(\opt \cup A_i)$ for all $i$. We set $\alpha = 0$ in each pass. In the first pass, we will set $\beta_1 = 1$. Then, since $S_0 = \emptyset$ Lemma~\ref{lem:single-pass} immediately gives:
    \begin{equation}
    \label{eq:first-pass}
    f(\opt) \leq f(\opt \cup A_1) \leq (2p-1)\left[f(S_1) - f(\emptyset)\right] + (2p+1)f(S_1) = 4pf(S_1)\,.
    \end{equation}
    For passes $i > 1$, we use the following, which relates the approximation guarantee obtained in this pass to that from the previous pass.
    \begin{theorem}
    \label{thm:guarantee-main}
    For $i > 1$, suppose that $f(\opt) \leq \gamma_{i-1}\cdot f(S_{i-1})$ and define $\delta_i = \frac{f(S_{i-1})}{f(S_{i})}$ as the ratio between the two previous passes. Then,
    \[
    f(\opt) \leq \min\left\{\gamma_{i-1}\delta_i, (\tfrac{p}{\beta_i} + p - 1)(1-\delta_i) + p + \beta_i p + 1\right\}\cdot f(S_{i}) + k \alpha\,.
    \]
    \end{theorem}
    \begin{proof}
    From the definition of $\gamma_{i-1}$ and $\delta_i$, we have:
    \begin{equation*}
    f(\opt) \leq \gamma_{i-1}f(S_{i-1}) = \gamma_{i-1}\delta_if(S_{i})\,.
    \end{equation*}
    On the other hand, $f(S_{i}) - f(S_{i-1}) = (1 - \delta_i)f(S_{i})$. Thus, Lemma~\ref{lem:single-pass} gives:
    \begin{equation*}
    f(\opt) \leq
    \left[\left(p/\beta_i + p - 1\right)(1 - \delta_i) + p + \beta_i p + 1\right]f(S_{i}) + k\alpha\,. \qedhere
    \end{equation*}
    \end{proof}
    Now, we observe that for any fixed guarantee $\gamma_{i-1}$ from the previous pass, $\gamma_{i-1}\delta_i$ is an increasing function of $\delta_i$ and $(p/\beta_i + p - 1)(1-\delta_i) + p + \beta_i p + 1$ is an decreasing function of $\delta_i$. Thus, the guarantee we obtain in Theorem~\ref{thm:guarantee-main} is always at least as good as that obtained when these two values are equal. Setting:
    \[
    \gamma_{i-1}\delta_i = (\tfrac{p}{\beta_i} + p - 1)(1-\delta_i) + p + \beta_i p + 1,
    \]
    and solving for $\delta_i$ gives us:
    \begin{equation}
    \label{eq:delta}
    \delta_i = \frac{p(1+\beta_i)^2}{p + \beta_i(\gamma_{i-1}-1+p)}\,.
    \end{equation}
    In the following analysis, we consider this value of $\delta_i$ since the guarantee given by Theorem~\ref{thm:guarantee-main} will always be no worse than that given by this value. The analysis for a single matroid constraint follows from our results for $p$-matchoids, but the analysis and parameter values obtained are much simpler, so we present it separately, first.
    \begin{theorem}
    \label{thm:monotone-matroid}
    Suppose we run Algorithm~\ref{alg:multi-pass} for an arbitrary matroid constraint and monotone submodular function $f$, with $\beta_i = \tfrac{1}{i}$. Then $2(1 + \tfrac{1}{i})f(S_i) \geq f(\opt)$ for all $i > 0$.
    In particular, after $i = \tfrac{2}{\e}$ passes, $(2+\e)f(S_i) \geq f(\opt)$.
    \end{theorem}
    \begin{proof}
    Let $\gamma_i$ be the guarantee for our algorithm after $i$ passes. We show, by induction on $i$, that $\gamma_i \leq \frac{2(i+1)}{i}$. For $i = 1$, we have $\beta_1 = 1$ and so from \eqref{eq:first-pass} we have $\gamma_1 = 4$, as required. For $i > 1$, suppose that $\gamma_{i-1} \leq \frac{2i}{i-1}$. Since $p = 1$ and $\beta_i = 1/i$, identity \eqref{eq:delta} gives:
    \begin{equation*}
    \delta_i \leq \frac{(1+\frac{1}{i})^2}{1 + \frac{1}{i}(\frac{2i}{i-1})}
    = \frac{\frac{(i+1)^2}{i^2}}{\frac{(i-1)+2}{i-1}}
    = \frac{(i-1)(i+1)}{i^2}\,.
    \end{equation*}

    Thus, by Theorem~\ref{thm:guarantee-main}, the $i^{\mathrm{th}}$ pass of our algorithm has guarantee $\gamma_i$ satisfying:
    \[
    \gamma_{i} \leq \gamma_{i-1}\delta_i \leq \frac{2i}{i-1}\frac{(i-1)(i+1)}{i^2} = \frac{2(i+1)}{i}\,,
    \]
    as required.
    \end{proof}

    \begin{theorem}\label{thm:monotone-p-matroid}
    Suppose we run Algorithm~\ref{alg:multi-pass} for an arbitrary $p$-matchoid constraint and monotone submodular function $f$, $\beta_1 = 1$ and
    \[\beta_i =\frac{\gamma_{i-1} - 1 - p}{\gamma_{i-1} - 1 + p},\]
    for $i > 1$, where $\gamma_i$ is given by the recurrence $\gamma_1 = 4p$ and
    \[
    \gamma_{i} =4p\frac{\gamma_{i-1}(\gamma_{i-1} - 1)}{(\gamma_{i-1} - 1 + p)^2},
    \]
    for $i > 1$. Then $\lb p + 1 + \tfrac{4p}{i} \rb f(S_i) \geq f(\opt)$ for all $i > 0$.
    In particular, after $i = \tfrac{4p}{\e}$ passes, $(p+1+\e)f(S_i) \geq f(\opt)$.
    \end{theorem}
    \begin{proof}
    We first show that approximation guarantee of our algorithm after $i$ passes is given by $\gamma_i$. Setting $\beta_1 = 1$, we obtain $\gamma_1=4p$ from \eqref{eq:first-pass}, agreeing with our definition. For passes $i > 1$, let $\beta_i= \frac{\gamma_{i-1}-1-p}{\gamma_{i-1}-1+p}$. As in the case of matroid constraint, Theorem~\ref{thm:guarantee-main} implies that the guarantee for pass $i$ will be at most $\delta_i\gamma_{i-1}$, where $\delta_i$ is chosen to satisfy \eqref{eq:delta}. Specifically, if we set
    \begin{equation*}
    \delta_i =
    \frac{p\left(1 +\frac{\gamma_{i-1}-1-p}{\gamma_{i-1}-1+p}\right)^2}{p + \frac{\gamma_{i-1}-1-p}{\gamma_{i-1}-1+p}(\gamma_{i-1}-1+p)}=
    \frac{p\left(\frac{2(\gamma_{i-1}-1)}{\gamma_{i-1}-1+p}\right)^2}{\gamma_{i-1}-1}
    =\frac{4p(\gamma_{i-1}-1)}{(\gamma_{i-1}-1+p)^2}\,,
    \end{equation*}
    then we have $\delta_i\gamma_{i-1} = \gamma_i$.

    We now show by induction on $i$ that $\gamma_i \leq p+1+\frac{4p}{i}$. In the case $i = 1$, we have $\gamma_1 = 4p$ and the claim follows immediately from $p \geq 1$. In the general case $i > 0$, and we may assume without loss of generality that $\gamma_{i-1} \geq 1$. Otherwise the theorem holds immediately, as each subsequent pass can only increase the value of the solution. Then, we note (as shown in Appendix~\ref{sec:calc-proof-theor}) that for $p \geq 1$ and $\gamma_{i-1} \geq 1$, $\gamma_i$ is an increasing function of $\gamma_{i-1}$. By the induction hypothesis, $\gamma_{i-1} \leq p+1+\frac{4p}{i-1}$. Therefore:
    \begin{equation*}
    \gamma_i \leq \frac{4p\left(p+1+\frac{4p}{i-1}\right)\left(p+\frac{4p}{i-1}\right)}{\left(2p + \frac{4p}{i-1}\right)^2} \leq p + 1 + \tfrac{4p}{i}\,,
    \end{equation*}
    as required. The last inequality above follows from straightforward but tedious algebraic manipulations, which can be found in Appendix~\ref{sec:calc-proof-theor}.
    \end{proof}



    \section{A multi-pass algorithm for general submodular \,\, \, functions}
    \label{sec:rand-multi-pass}
    In this section, we show that the guarantees for monotone submdodular maximization can be extended to non-monotone submodular maximization even when dealing with multiple passes. Our main algorithm is given by procedure \mrSLS{} in Algorithm~\ref{alg:multi-pass-random}. In each pass, it calls a procedure \rSLS{}, which is an adaptation of \SLS{}, to process the stream.
    Note that each such pass produces a pair of feasible solutions $S$ and $S'$, which we now maintain throughout \mrSLS{}. The set $S$ is maintained similarly as before and gradually improves by exchanging ``good'' elements into a solution throughout the pass. The set $S'$ will be maintained by considering the best output of an offline algorithm that we run after each pass as described in more detail below.
    \begin{algorithm}[t]
    \SetKwProg{myproc}{procedure}{}{}
    \myproc{$\textsc{MultipassRandomizedLocalSearch}(\alpha, \beta_1,\ldots,\beta_d,m)$}{
    $S_0 \gets \emptyset, S'_0 \gets \emptyset$\;
    \For{$i = 1$ \KwTo $d$}{
      Let $(\tilde{S}, S')$ be the output of $\textsc{RandomizedLocalSearch}(S_{i-1},\alpha, \beta_i,m)$\;
      $ S_i\gets \tilde{S}$, $S'_i \gets \argmax \{f(S'_{i-1}), f(S')\}$\;
    }
    \Return $\bar{S} = \argmax \{ f(S_d), f(S'_d) \}$\;
    }
    \BlankLine
    \myproc{$\textsc{RandomizedLocalSearch}(\Sinit,\alpha,\beta,m)$}{
    $S \gets \Sinit$; $B \gets \emptyset$\;
    \ForEach{$x$ in the stream}
    {
      \If{$f(x\mid S) \geq \alpha + (1+\beta) \sum_{e \in C_x} \nu(e, S)$}
      {
      $B \gets B + x$\;
      }
      \If{ $\card{B} = m$}
      {
      $x \gets$ uniformly random element from $B$\;
      $C_x \gets \textsc{Exchange}(x, S)$\;
      $B \gets B - x$; $S \gets S + x - C_x$\;
        \ForEach{$x'$ in $B$}
        {
          $C_{x'} \gets \textsc{Exchange}(x',S)$\;
          \If{$f(x'\mid S) < \alpha + (1+\beta) \sum_{e \in C_{x'}} \nu(e, S)$}
          {
          $B \gets B - x'$\;
          }
        }
      }
    }
    $S' \gets \textsc{Offline}(B)$\;
    \Return $(S, S')$\;
    }
    \caption{The randomized multi-pass streaming algorithm}
    \label{alg:multi-pass-random}
    \end{algorithm}

    To deal with non-monotone submodular functions, we will limit the probability of elements being added to $S$. Instead of exchanging good elements on arrival, we store them in a buffer $B$ of size $m$. When the buffer becomes full, an element is chosen uniformly at random and added to $S$.
    Adding a new element to the current solution may affect the quality of the remaining elements in the buffer and thus we need to re-evaluate them and remove the elements that are no longer good. As before, we let $A$ denote the set of elements that were previously added to $S$ during the current pass of the algorithm. Note that we do not consider an element to be accepted until it has actually been added to $S$ from the buffer. For any fixed set of random choices, the execution of \rSLS{} can be considered as the execution of \SLS{} on the following stream: we suppose that an element $x$ arrives whenever it is selected from the buffer and accepted into $S$. All elements that are discarded from the buffer after accepting $x$ then arrive, and will also be rejected by \SLS{}. Any elements remaining in the buffer after the execution of the algorithm do not arrive in the stream. Applying Lemma~\ref{lem:single-pass} with respect to this pretend stream ordering allows us to bound $f(\tilde{S})$ with respect to  $f(\opt \setminus B)$ (that is, the value of the part of $\opt$ that does not remain in the buffer $B$) after a single pass of \rSLS. Formally, let $\tilde{B}_i$ be the value of the buffer after the $i^{\mathrm{th}}$ pass of our algorithm. Then, applying Lemma~\ref{lem:single-pass} to the set $\opt \bb \tilde{B}_i$, and taking expectation, gives:
    \begin{multline}
        \esp{f(A_{i} \cup (\opt \bb \tilde{B}_i))}
        \leq \lb p/\beta + p - 1 \rb \left( \esp{ f( S_i
     ) } - \esp{f ( S_{i-1} )} \right) \\
      \qquad + (p + \beta p + 1) \esp{f( S_{i})} + \alpha k\,.
    \label{eq:opt-1}
    \end{multline}
    In order to bound the value of the elements in $\tilde{B}_i$, we apply any offline $\gammaoff$-approximation algorithm \textsc{Offline} to the buffer at the end of the pass to obtain a solution $S'$. In \mrSLS{}, we then remember the best such offline solution $S'_i$ computed across the first $i$ passes. Then, in the $i^{\mathrm{th}}$ pass, we have
     \begin{equation}
    \label{eq:opt-2}
    \esp{f(\opt \cap \tilde{B}_i)} \leq \gammaoff \esp{f(S')} \leq \gammaoff \esp{f(S'_i)}\,.
    \end{equation}
    From submodularity of $f$ and $A_i \cap \tilde{B}_i = \emptyset$ we have $f(A_{i} \cup \opt) \leq f(A_i \cup (\opt \bb \tilde{B}_i)) + f(\opt \cap \tilde{B}_i)$. Thus, combining \eqref{eq:opt-1} and \eqref{eq:opt-2} we have:
    \begin{multline}
    \label{eq:opt-3}
    \esp{f(A_{i} \cup \opt)}
        \leq (p/\beta+ p - 1)\left( \esp{ f( S_i
     ) } - \esp{f ( S_{i-1} )} \right) \\
    + (p + \beta p + 1) \esp{f( S_{i})} + \gammaoff \esp{f(S'_i)} + \alpha k\,.
    \end{multline}
    To relate the right-hand side to $f(\opt)$ we use the following result
    from Buchbinder et al.~\cite{DBLP:conf/soda/BuchbinderFNS14}:
    \begin{lemma}[Lemma 2.2 in {\cite{DBLP:conf/soda/BuchbinderFNS14}}]
    Let $f \colon 2^X \rightarrow \RR_{\geq 0}$ be a non-negative submodular function. Suppose that $A$ is a random set where no element $e \in X$ appears  in $A$ with probability more than $p$.
    Then, $\esp{f(A)} \geq \lb 1 - p \rb f\lb \emptyset\rb$. Moreover, for any set $Y \subseteq X$, it follows that $\esp{f\lb Y \cup A\rb\,} \geq (1-p) {f(Y)}$.
    \label{lem:thm3}
    \end{lemma}
    We remark that a similar theorem also appeared earlier in Feige, Mirrokni, and Vondr\'ak~\cite{DBLP:journals/siamcomp/FeigeMV11} for a random set that contains each element \emph{independently} with probability \emph{exactly} $p$. Here, the probability that an element occurs in $A_i$ is delicate to handle because such an element may either originate from the starting solution $S_{i-1}$ or be added during the pass. Thus, we use a rougher estimate.
    By definition $ A_{i} \subseteq A_{i} \cup A_{i-1}\cup \ldots \cup A_1$. Thus,
    $\prob\ld e \in A_{i} \rd \leq \prob\ld e \in A_{i} \cup \ldots \cup A_1 \rd$.
    The number of selections during the $j^\textrm{th}$ pass is at most $\card{A_j}$ and by Lemma~\ref{lem:single-pass} (applied to the set $\opt \setminus \tilde{B}_j$ due to our pretend stream ordering in each pass $j$), $\card{A_j} \leq f(\opt \setminus \tilde{B_j})/\alpha \leq f(\opt)/\alpha$ in any pass. Here, the second inequality follows from the optimality of $\opt$, and the fact that any subset of the feasible solution $\opt$ is also feasible for our $p$-matchoid constraint. Thus, the total number of selections in the first $i$ passes at most $\sum_{j=1}^i \card{A_j} \leq i\cdot f(\opt)/\alpha$. We select an element only when the buffer is full, and each selection is made independently and uniformly at random from the buffer. Thus, the probability that any given element is selected when the algorithm makes a selection is at most $1/m$ and by a union bound, $\prob\ld e \in A_{i} \cup \ldots \cup A_1 \rd \leq i\cdot f(\opt)/(m\alpha)$. Let $d$ be the number of passes that the algorithm makes and suppose we set $\alpha = \e f(\opt)/2k$ (in Appendix~\ref{sec:analys-non-monot} we show that this can be accomplished approximately by guessing $f(\opt)$, which can be done at the expense of an extra factor $O(\log k)$ space). Finally, let $m = 4d k /\e^2$. Then, applying Lemma~\ref{lem:thm3}, after $i \leq d$ passes we have:
    \begin{equation}
    \esp{f(A_i \cup \opt)} \geq \lb 1 - d \cdot f(\opt)/(m \alpha)\rb f(\opt) \geq \lb 1 - \e/2 \rb f(\opt)\,.
    \label{eq:non-monotone-bound1}
    \end{equation}
    Our definition of $\alpha$ also implies that $\alpha k \leq \e /2f(\opt)$. Using this and equation \eqref{eq:non-monotone-bound1} in \eqref{eq:opt-3}, we obtain:
    \begin{multline}
    (1 - \e)f(\opt) \\
    \leq (p/\beta + p -1)(\esp{f(S_i)} - \esp{f(S_{i-1})})+(p+\beta p + 1)\esp{f(S_i)} + \gammaoff\esp{f(S'_i)}\,.\label{eq:non-monotone-main}
    \end{multline}
    As we show in Appendix~\ref{sec:analys-non-monot}, the rest of the analysis then follows similarly to that in Section~\ref{sec:select-param-beta_i}, using the fact that $f(\bar{S}) = \max\{f(S_d),f(S'_d)\}$.
    \begin{restatable}{theorem}{nonmonmain} Let $\cM^p = \lb X, \cI \rb$  be a p-matchoid of rank $k$ and let $f \colon 2^X \rightarrow \posreals$ be a non-negative submodular function. Suppose there exists an algorithm for the offline instance of the problem with approximation factor $\gammaoff$. For any $\e > 0$, the randomized streaming local-search algorithm returns a solution $\bar{S} \in \cI$ such that
    \begin{equation*}
    \fopt \leq \lb p +1 + {\gammaoff} + \bigO{\e} \rb \esp{f(\bar{S})}
    \end{equation*}
    using a total space of $\bigO{\frac{p^3 k\log_2 k}{\e^3}}$ and $\bigO{\frac{p}{\e}}$-passes.
    \label{thm:non-mono-matroid}
    \end{restatable}

    \bibliographystyle{abbrv}
    \bibliography{submodular}

\begin{thebibliography}{10}

\bibitem{Badanidiyuru:2014ib}
A.~Badanidiyuru, B.~Mirzasoleiman, A.~Karbasi, and A.~Krause.
\newblock Streaming submodular maximization: massive data summarization on the
  fly.
\newblock In S.~A. Macskassy, C.~Perlich, J.~Leskovec, W.~Wang, and R.~Ghani,
  editors, {\em The 20th {ACM} {SIGKDD} International Conference on Knowledge
  Discovery and Data Mining, {KDD} '14, New York, NY, {USA} - August 24 - 27,
  2014}, pages 671--680. {ACM}, 2014.

\bibitem{Badanidiyuru:2013jc}
A.~Badanidiyuru and J.~Vondr{\'a}k.
\newblock Fast algorithms for maximizing submodular functions.
\newblock In {\em Proc. ACM-SIAM Symposium on Discrete Algorithms (SODA)},
  pages 1497--1514, 2013.

\bibitem{Buchbinder:2016}
N.~Buchbinder and M.~Feldman.
\newblock Constrained submodular maximization via a nonsymmetric technique.
\newblock {\em Mathematics of Operations Research}, 44(3):988--1005, 2019.

\bibitem{DBLP:conf/soda/BuchbinderFNS14}
N.~Buchbinder, M.~Feldman, J.~Naor, and R.~Schwartz.
\newblock Submodular maximization with cardinality constraints.
\newblock In {\em Proc. {ACM-SIAM} Symposium on Discrete Algorithms (SODA)},
  pages 1433--1452, 2014.

\bibitem{Calinescu:2011ju}
G.~Calinescu, C.~Chekuri, M.~P{\'a}l, and J.~Vondr{\'a}k.
\newblock Maximizing a monotone submodular function subject to a matroid
  constraint.
\newblock {\em SIAM Journal on Computing}, 40(6):1740--1766, 2011.

\bibitem{DBLP:journals/mp/ChakrabartiK15}
A.~Chakrabarti and S.~Kale.
\newblock Submodular maximization meets streaming: matchings, matroids, and
  more.
\newblock {\em Mathematical Programming}, 154(1-2):225--247, 2015.

\bibitem{DBLP:conf/icalp/ChekuriGQ15}
C.~Chekuri, S.~Gupta, and K.~Quanrud.
\newblock Streaming algorithms for submodular function maximization.
\newblock In M.~M. Halld{\'{o}}rsson, K.~Iwama, N.~Kobayashi, and B.~Speckmann,
  editors, {\em Automata, Languages, and Programming - 42nd International
  Colloquium, {ICALP} 2015, Kyoto, Japan, July 6-10, 2015, Proceedings, Part
  {I}}, volume 9134 of {\em Lecture Notes in Computer Science}, pages 318--330.
  Springer, 2015.

\bibitem{DBLP:journals/siamcomp/ChekuriVZ14}
C.~Chekuri, J.~Vondr{\'{a}}k, and R.~Zenklusen.
\newblock Submodular function maximization via the multilinear relaxation and
  contention resolution schemes.
\newblock {\em {SIAM} Journal on Computing}, 43(6):1831--1879, 2014.

\bibitem{Feige1998}
U.~Feige.
\newblock A threshold of $\ln n$ for approximating set cover.
\newblock {\em Journal of the ACM}, 45(4):634--652, 1998.

\bibitem{DBLP:journals/siamcomp/FeigeMV11}
U.~Feige, V.~S. Mirrokni, and J.~Vondr{\'{a}}k.
\newblock Maximizing non-monotone submodular functions.
\newblock {\em SIAM Journal on Computing}, 40(4):1133--1153, 2011.

\bibitem{Feldman2018DoLess}
M.~Feldman, A.~Karbasi, and E.~Kazemi.
\newblock Do less, get more: streaming submodular maximization with
  subsampling.
\newblock In {\em Advances in Neural Information Processing Systems (NeurIPS)},
  pages 732--742, 2018.

\bibitem{DBLP:conf/focs/FeldmanNS11}
M.~Feldman, J.~Naor, and R.~Schwartz.
\newblock A unified continuous greedy algorithm for submodular maximization.
\newblock In {\em Proc. IEEE Symposium on Foundations of Computer Science,
  (FOCS)}, pages 570--579, 2011.

\bibitem{Feldman2011}
M.~Feldman, J.~S. Naor, R.~Schwartz, and J.~Ward.
\newblock Improved approximations for k-exchange systems.
\newblock In {\em Proc. European Symposium on Algorithms (ESA)}, pages
  784--798, 2011.

\bibitem{Feldman2020}
M.~Feldman, A.~Norouzi-Fard, O.~Svensson, and R.~Zenklusen.
\newblock The one-way communication complexity of submodular maximization with
  applications to streaming and robustness.
\newblock In {\em Proc. ACM Symposium on Theory of Computing (STOC)}, pages
  1363--1374, 2020.

\bibitem{FisherNemhauserWolsey}
M.~L. Fisher, G.~L. Nemhauser, and L.~A. Wolsey.
\newblock An analysis of approximations for maximizing submodular set functions
  {II}.
\newblock {\em Mathematical Programming Study}, 8:73--87, 1978.

\bibitem{Gharan:2011}
S.~O. Gharan and J.~Vondr{\'a}k.
\newblock Submodular maximization by simulated annealing.
\newblock In {\em Proc. ACM-SIAM Symposium on Discrete Algorithms (SODA)},
  pages 1098--1116, 2011.

\bibitem{HuangKakimuraMultiPass2018}
C.-C. Huang and N.~Kakimura.
\newblock Multi-pass streaming algorithms for monotone submodular function
  maximization.
\newblock {\em CoRR}, abs/1802.06212, 2018.

\bibitem{Lee:2010}
J.~Lee, M.~Sviridenko, and J.~Vondr{\'a}k.
\newblock Submodular maximization over multiple matroids via generalized
  exchange properties.
\newblock {\em Mathematics of Operations Research}, 35(4):795--806, 2010.

\bibitem{Levin:2020:Streaming}
R.~Levin and D.~Wajc.
\newblock Streaming submodular matching meets the primal-dual method.
\newblock {\em arXiv preprint arXiv:2008.10062}, 2020.

\bibitem{lin:2010wpa}
H.~Lin and J.~A. Bilmes.
\newblock Multi-document summarization via budgeted maximization of submodular
  functions.
\newblock In {\em Human Language Technologies: Conference of the North American
  Chapter of the Association of Computational Linguistics, Proceedings, June
  2-4, 2010, Los Angeles, California, {USA}}, pages 912--920. The Association
  for Computational Linguistics, 2010.

\bibitem{mcGregor2017}
A.~McGregor and H.~T. Vu.
\newblock Better streaming algorithms for the maximum coverage problem.
\newblock {\em Theory of Computing Systems}, 63(7):1595--1619, 2019.

\bibitem{mirzasoleiman2016fast}
B.~Mirzasoleiman, A.~Badanidiyuru, and A.~Karbasi.
\newblock Fast constrained submodular maximization: Personalized data
  summarization.
\newblock In {\em ICML}, pages 1358--1367, 2016.

\bibitem{Mirzasoleiman:2015}
B.~Mirzasoleiman, A.~Badanidiyuru, A.~Karbasi, J.~Vondr\'{a}k, and A.~Krause.
\newblock Lazier than lazy greedy.
\newblock In {\em Proc. AAAI Conference on Artificial Intelligence (AAAI)},
  pages 1812--1818, 2015.

\bibitem{mirzasoleiman18streaming}
B.~Mirzasoleiman, S.~Jegelka, and A.~Krause.
\newblock Streaming non-monotone submodular maximization: Personalized video
  summarization on the fly.
\newblock {\em arXiv preprint arXiv:1706.03583}, 2017.

\bibitem{Nemhauser:1978dm}
G.~L. Nemhauser and L.~A. Wolsey.
\newblock Best algorithms for approximating the maximum of a submodular set
  function.
\newblock {\em Mathematics of Operations Research}, 3(3):177--188, 1978.

\bibitem{Ashkan2018}
A.~Norouzi-Fard, J.~Tarnawski, S.~Mitrovi{\'c}, A.~Zandieh, A.~Mousavifar, and
  O.~Svensson.
\newblock Beyond 1/2-approximation for submodular maximization on massive data
  streams.
\newblock In {\em Proc. International Conference on Machine Learning (ICML)},
  pages 3826--3835, 2018.

\end{thebibliography}


    \appendix

    \section{Proof of Lemma~\ref{lem:opt-R}}
    \label{sec:proof-lemma-refl}
    Here, we give a self-contained analysis of the single-pass algorithm of Chekuri, Gupta, and Quanrud~\cite{DBLP:conf/icalp/ChekuriGQ15}, corresponding to Algorithm~\ref{alg:multi-pass} initialized with $\Sinit = \emptyset$. First, we prove Lemma~\ref{lem:incremental-properties}, which concerns properties of the incremental values maintained by Algorithm~\ref{alg:multi-pass}.
    \label{sec:analys-algor-refalg}
    \incremental*
    \begin{proof}
    Property (1) follows directly from the telescoping summation
    \begin{equation*}
    \sum_{e\in T}\nu(e,T) = \sum_{e \in T}[f(e \cup \{t' \in T : t' \prec e\}) - f(\{t' \in T : t' \prec e\}] = f(T) - f(\emptyset).
    \end{equation*}

    Property (2) follows from submodularity since $T\subseteq U$ implies that $\{t' \in T:t' \prec e\} \subseteq \{t' \in U : t' \prec e\}$.

    For property (3), we note that:
    \begin{align*}
    f(T \mid U \bb T) & = \sum_{t \in T}f(t \mid U \bb T \cup \{t' \in T:t' \prec t\}), \\
    &\leq \sum_{t \in T} f(t \mid \{u' \in U:u' \prec t\}),\\
    & = \sum_{t \in T} \nu(t, U)\,,
    \end{align*}
    where the first equation follows from a telescoping summation, and the inequality follows from submodularity, since $\{u' \in U : u' \prec t\} \subseteq U \setminus T \cup \{t'\in T:t' \prec t\}$.

    We prove property (4) by induction on the stream of elements arriving. Initially $S = \emptyset$. Thus, the first time that any element $x$ is accepted, we must have $C_x = \emptyset$ and so $f(x \mid S) \geq \alpha \geq 0$. After this element is accepted, we have $\nu(x,S) = \nu(x,\{x\}) = f(x \mid \emptyset) = \alpha$. Proceeding inductively, then, let $S_x^-$ and $S_x^+$ be the set of elements in $S$ before and after some new element $x$ arrives and is processed by Algorithm~\ref{alg:multi-pass}, and suppose that $\nu(s,S_x^-) \geq \alpha$ for all $s \in S_x^-$. Then, if $x$ is rejected, we have $S_x^+ = S_x^-$ and so $\nu(s,S_x^+) = \nu(s,S_x^-) \geq \alpha$ for all $s \in S_x^+$. If $x$ is accepted, then $S_x^+ = S \bb C_x + x$ and $f(x\mid S_x^-) \geq \alpha + (1+\beta)\sum_{e \in C_x}\nu(e,S_x^-)$. Thus,
    \begin{equation*}
    \nu(x,S_x^+) \geq f(x \mid S_x^+ - x) \geq f(x  \mid S_x^-) \geq \alpha + (1+\beta)|C_x|\alpha \geq \alpha\,,
    \end{equation*}
    where the first inequality follows from property (2) of the lemma, the second from submodularity, and the third from the induction hypothesis and the assumption that $x$ is accepted. For any other $s \in S_x^+$, we have $\{t' \in S \bb C_x : t' \prec s\} \subseteq \{t' \in S : t' \prec s\}$ and so by property (3) of the lemma, $\nu(s, S_x^+) \geq \nu(s, S_x^-) \geq \alpha$, as required.
    \end{proof}

    In our analysis we will use the following structural lemma from Chekuri et al.~\cite{DBLP:conf/icalp/ChekuriGQ15} (here, restated in our notation). This lemma applies to the execution of our algorithm \SLS{} when $\Sinit = \emptyset$, and so no element is discarded upon arrival due to $x \in \Sinit$. However, we note that the execution of our algorithm is in fact exactly the same as this algorithm executed on the pretend stream ordering introduced in Section~\ref{sec:multi-pass-main} to define the incremental values $\nu$. Specifically, in each pass of our algorithm, the set $\Sinit$ is a feasible solution produced by the preceding pass and in the pretend stream ordering, all elements of $\Sinit$ arrive in our pretend ordering in the same relative (pretend) order as this preceding pass. It follows that whenever $x \in \Sinit$ arrives in our pretend ordering for the present pass, we have $C_x = \emptyset$ and $\nu(x,S) = \nu(x,\Sinit) \geq \alpha$ by Lemma~\ref{lem:incremental-properties}~(4), since $x$ was present in the feasible solution $S = \Sinit$ at the end of the preceding pass. Thus, each $x \in \Sinit$ will first be accepted in our pretend stream ordering, and then the rest of  $X \bb \Sinit$ is processed, exactly as in \SLS{}.

    Recall that we let $\tilde{A}$ be the set of all elements that were accepted by this pass of \SLS{} (and so at some point appeared in $S$). For each element $x \in X$, we let $S_x^-$ be the current set $S$ at the moment that $x$ arrives and $S_x^+$ the set after $x$ is processed. For an element $e$ that is accepted but later \emph{evicted} from $S$, let $\chi(e)$ be the incremental value $\nu(e,S)$ of $e$ at the moment that $e$ was evicted.
    \begin{lemma}[Lemma 9 of {\cite{DBLP:conf/icalp/ChekuriGQ15}}]
    \label{lem:cgq}
    Let $T \in \cI$ be a feasible solution disjoint from $\tilde{A}$, and $\tilde{S}$ be the output of the streaming algorithm. There exists a mapping $\varphi : T \to 2^{\tilde{A}}$ such that:
    \begin{enumerate}
    \item Every $s \in \tilde{S}$ appears in the set $\varphi(t)$ for at most $p$ choices of $t \in T$.
    \item Every $e \in \tilde{A} \bb \tilde{S}$ appears in the set $\varphi(t)$ for at most $p-1$ choices of $t \in T$.
    \item For each $t \in T$:
    \[
    \sum_{c \in C_t}\nu(c,S_t^-) \leq \sum_{e \in \varphi(t) \bb \tilde{S}}\chi(e) + \sum_{s \in \varphi(t) \cap \tilde{S}}\nu(s,\tilde{S})\,.
    \]
    \end{enumerate}
    \end{lemma}

    Using this charging argument, we can now prove Lemma~\ref{lem:opt-R} directly.
    \onepass*
    \begin{proof}
    Let $R = \opt \bb \tilde{A}$. Since $S^-_r \subseteq \tilde{A}$ for all $r$, the submodularity of $f$ implies that
    \begin{equation}
    \label{eq:opt-r}
    \sum_{r \in R} f(r \mid S_r^-) \geq \sum_{r \in R}f(r \mid \tilde{A})
    \geq f(R \cup \tilde{A}) - f(\tilde{A})
    = f(\opt \cup \tilde{A}) - f(\tilde{A})\,.
    \end{equation}
    For any $r \in R$, since $r$ was rejected upon arrival,
    \begin{equation}
    \label{eq:rejected}
    f(r\mid S_r^-) \leq (1+\beta) \sum_{c \in C_r}\nu(c,S_r^-) + \alpha\,.
    \end{equation}
    Thus, applying Lemma~\ref{lem:cgq} we obtain:
    \begin{align*}
    \sum_{r \in R}&f(r \mid S_r^-)
    \leq (1 + \beta)\sum_{r \in R}\sum_{c \in C_r}\nu(c,S_r^-) + k\alpha, &
    \text{(\eqref{eq:rejected} and $|R| \leq k$)}&
    \\
    &\leq \sum_{r \in R}(1+\beta)\biggl[\sum_{e \in \varphi(r) \bb \tilde{S}}\chi(e) + \sum_{s \in \varphi(r) \cap \tilde{S}}\nu(s,\tilde{S})\biggr] + k\alpha, \hspace{-1em} &\text{(Lemma \ref{lem:cgq} (3))}& \\
    &\leq (1+\beta)\biggl[(p - 1)\sum_{e \in \tilde{A}\bb\tilde{S}}\chi(e) + p \sum_{s \in \tilde{S}}\nu(s,\tilde{S})\biggr] + k\alpha, &\text{(Lemma \ref{lem:cgq} (1, 2))}&
    \end{align*}
    where in the last inequality we have also used Lemma~\ref{lem:incremental-properties}~(4), which implies that each $\chi(e)$ and $\nu(s,\tilde{S})$ is non-negative. Combining the above inequality with \eqref{eq:opt-r}, we obtain
    \begin{equation}
    \label{eq:main-1}
    f(\opt \cup \tilde{A}) \leq (1+\beta)\left[(p - 1)\sum_{e \in \tilde{A}\bb\tilde{S}}\chi(e) + p\sum_{s \in \tilde{S}}\nu(s,\tilde{S})\right]  + f(\tilde{A}) + k\alpha\,.
    \end{equation}

    We now bound $f(\tilde{A})$ in terms of the values $\nu(s,\tilde{S})$ and $\chi(e)$. Since $S \subseteq \tilde{A}$ at all times during the algorithm, and $\chi(e) = \nu(e, S)$ at the moment $e$ was evicted, we have $\chi(e) \geq \nu(e, \tilde{A})$ by Lemma~\ref{lem:incremental-properties}~(2). Thus,
    \begin{equation}
    \label{eq:tilde-a-bound}
    f(\tilde{A}) - f(\emptyset) = \sum_{a \in \tilde{A}}\nu(a, \tilde{A})
    = \sum_{s \in \tilde{S}} \nu(s, \tilde{A}) + \sum_{e \in \tilde{A} \bb \tilde{S}}\nu(e, \tilde{A})
    \leq \sum_{s \in \tilde{S}}\nu(s, \tilde{S}) + \sum_{e \in \tilde{A} \bb \tilde{S}} \chi(e)\,,
    \end{equation}
    where the first equation follows from Lemma~\ref{lem:incremental-properties}~(1), and the last inequality follows from Lemma~\ref{lem:incremental-properties}~(2).

    Combining \eqref{eq:main-1} and \eqref{eq:tilde-a-bound} we have:
    \begin{align}
    f(\opt\cup \tilde{A}) &\leq \left((1+\beta)(p - 1) + 1\right)\sum_{e \in \tilde{A}\bb\tilde{S}}\chi(e) \notag \\
    & \qquad + \left((1+\beta)p + 1\right)\sum_{e \in \tilde{S}}\nu(s,\tilde{S}) + f(\emptyset) + k\alpha, \notag \\
    &= (p + p\beta - \beta)\sum_{e \in \tilde{A}\bb\tilde{S}}\!\chi(e) +
    (p+\beta p + 1)\!\sum_{s \in \tilde{S}}\nu(s,\tilde{S}) + f(\emptyset) + k\alpha\,.
    \label{eq:lem-2-main}
    \end{align}
    By Lemma~\ref{lem:incremental-properties}~(1), we have the following bound for the second summation in \eqref{eq:lem-2-main}:
    \begin{equation*}
    (p+\beta p + 1)\sum_{e \in \tilde{S}}\nu(e,\tilde{S}) + f(\emptyset) = (p+\beta p + 1)[f(\tilde{S}) - f(\emptyset)] + f(\emptyset) \leq (p + \beta p + 1)f(\tilde{S})\,.
    \end{equation*}
    Combining this and \eqref{eq:lem-2-main} we obtain:
    \begin{equation*}
    f(\opt \cup \tilde{A}) \leq (p + p\beta - \beta)\sum_{e \in \tilde{A}\bb\tilde{S}}\chi(e) +
    (p+\beta p + 1)f(\tilde{S}) + k\alpha.\qedhere
    \end{equation*}
    \end{proof}

    \section{Calculations for the proof of Theorem~\ref{thm:monotone-p-matroid}}
    \label{sec:calc-proof-theor}
    We recall that
    \[\gamma_i = \gamma_{i-1}\delta_i = \frac{4p\gamma_{i-1}(\gamma_{i-1}-1)}{(\gamma_{i-1}-1+p)^2}\,.
    \]
    Then, to see that $\gamma_i$ is an increasing function of $\gamma_{i-1}$ for $p \geq 1$ and $\gamma_{i-1} \geq 1$, we note that:
    \begin{align*}
    \frac{d}{d\gamma_{i-1}}\gamma_i
    &=
    \frac{4p(\gamma_{i-1} - 1) + 4p\gamma_{i-1}}{(\gamma_{i-1} - 1 + p)^2}
    - \frac{8p\gamma_i(\gamma_{i-1} - 1)}{(\gamma_{i-1} - 1 + p)^3} \\
    & = \frac{4p(\gamma_{i-1} - 1)(\gamma_{i-1} - 1 + p)
    + 4p\gamma_{i-1}(\gamma_{i-1} - 1 + p) - 8p\gamma_{i-1}(\gamma_{i-1} - 1)}{(\gamma_{i-1} - 1 + p)^3}\! \\
    & \geq  \frac{4p\gamma_{i-1}(\gamma_{i-1} - 1)
    + 4p\gamma_{i-1}^2 - 8p\gamma_{i-1}(\gamma_{i-1} - 1)}{(\gamma_{i-1} - 1 + p)^3}\geq 0.
    \end{align*}
    The third line follows from $p \geq 1$ and the final inequality is by $\gamma_{i-1} \geq 1$.

    We now verify the following inequality used at the end of Theorem~\ref{thm:monotone-p-matroid}:
    \begin{equation*}
    \frac{4p\left(p+1+\frac{4p}{i-1}\right)\left(p+\frac{4p}{i-1}\right)}{\left(2p + \frac{4p}{i-1}\right)^2} \leq p + 1 + \tfrac{4p}{i}\,.
    \end{equation*}
    Rearranging both sides and placing over a common denominator gives:
    \begin{align*}
    \frac{4p\left(p+1+\frac{4p}{i-1}\right)\left(p+\frac{4p}{i-1}\right)}{\left(2p + \frac{4p}{i-1}\right)^2}
    &= \frac{4p\left((p+1)(i-1)+4p\right)\left(p(i-1) + 4p\right)}{\left(2p(i-1) + 4p\right)^2}, \\
    &= \frac{4p\left((p+1)(i-1)+4p\right)\left(p(i-1) + 4p\right)}{\left(2p(i+1)\right)^2}, \\
    &= \frac{\left((i-1)(p+1)+4p\right)(i+3)}{(i+1)^2}, \\
    &= \frac{(i-1)(i+3)i(p+1)+i(i+3)4p}{i(i+1)^2},\\
    &= \frac{\left(i^2 + 2i - 3\right)i(p+1)+(i^2+3i)4p}{i(i+1)^2},\\
    \intertext{and}
    p+1+\tfrac{4p}{i} &= \frac{(p+1)i + 4p}{i}, \\
    &= \frac{i(i+1)^2(p+1)+(i+1)^24p}{i(i+1)^2}, \\
    &= \frac{\left(i^2 + 2i + 1\right)i(p+1)+\left(i^2 + 2i + 1\right)4p}{i(i+1)^2}\,.
    \end{align*}
    Then, since $p\geq 1$ and $i \geq 1$,
    \begin{equation*}
    \left(p+1+\tfrac{4p}{i}\right)-\frac{4p\left(p+1+\frac{4p}{i-1}\right)\left(p+\frac{4p}{i-1}\right)}{\left(2p + \frac{4p}{i-1}\right)^2} = \frac{4i(p+1) - 4(i-1)p}{i(i+1)^2} \geq 0.
    \end{equation*}

    \section{Additional Details for the Non-Monotone Case}
    \label{sec:analys-non-monot}

    \subsection{Guessing the value of $f(\opt)$}
    \label{sec:guessing-value-opt}
    Guessing the value of $f(\opt)$ is a common technique in streaming submodular function maximization. Badanidiyuru et al.~\cite{Badanidiyuru:2014ib} showed how to approximate $f(\opt)$ within a constant factor using $\bigO{\log(k)}$ space in a single pass. To avoid extra complications, we show how to guess $f(\opt)$ in two passes and refer the reader to~\cite{Badanidiyuru:2014ib} for an approximation of $f(\opt)$ on the fly. Let $\tau = \max_{e \in X} f(e)$.
    Using submodularity, it is easy to see that $\tau \leq f(\opt) \leq k \tau$. Consider the set
    \begin{equation*}
    \Lambda = \lc 2^{i} \mid i \in \ints, \;  \tau \leq 2^{i} \leq k\cdot \tau \rc.
    \end{equation*}
    Then there exists a value $\lambda \in \Lambda$ such that $\frac{\fopt}{2} \leq \lambda \leq \fopt$. Setting the parameter $\alpha = \e\lambda/(2k)$, we get that $\alpha \in \ld \e\fopt/4k; \e\fopt/2k\rd$. The defined range of $\alpha$ is sufficient for the analysis\footnote{Equation \eqref{eq:non-monotone-bound1} and the bound $\alpha k \leq \e f(\opt)$ are where we need the exact value of $\alpha$, using upper and lower bounds for $\alpha$ yield the same result up to the hidden constant in the term $O(\epsilon)$.}.
    Unfortunately, it is still not possible to know which $\lambda \in \Lambda$ satisfies the property. However, it suffices to run the randomized local-search algorithm for every $\lambda \in \Lambda$ in parallel and output the best solution of all the copies. This operation increases the space complexity by a multiplicative $\bigO{\log_2 k}$ factor, and adds one additional pass to find $\tau$.

    \subsection{Proof of Theorem~\ref{thm:non-monotone-approximation}}
    \label{sec:proof-theorem}
    Here we give a full proof of the following theorem from Section~\ref{sec:rand-multi-pass}:
    \nonmonmain*
    In the same spirit as in Section~\ref{sec:select-param-beta_i}, we show that we can derive a guarantee with respect to the solution $\esp{f(S_i)}$ produced after the $i^\textrm{th}$ pass even when the function is non-monotone. In fact, we show that the analysis of the non-monotone case reduces to the monotone case as shown in the following theorem.
    \begin{theorem}\label{thm:non-mono-upper-bound}
    Let $f$ be a non-negative submodular function. Let the additive threshold $\alpha = \e f(\opt)/2k$ and let $d \geq i>1$. Suppose that at the start of the $i^\textrm{th}$ iteration of the randomized local-search algorithm with a buffer of size $m = 4 d k /\e^2$ we have $( 1 - \e) \fopt\! \leq\! \gamma_{i-1} \esp{f(S_{i-1})} + \gammaoff \esp{f(S'_{i-1})}$. Then,
    \begin{multline*}
        (1- \e )\fopt \leq \min \lc \gamma_{i-1} \delta_i, \lb \frac{p}{\beta_i} + p - 1 \rb \lb 1 - \delta_i \rb + p + \beta_i p + 1 \rc \cdot \esp{f(S_i)} \\
         \qquad + \gammaoff \esp{f(S'_i)},
    \end{multline*}
    where $\delta_i = \frac{\esp{f(S_{i-1})}}{\esp{f(S_{i})}}$.
    \end{theorem}
    \begin{proof}
    From the definition of $\gamma_{i-1}$ and $\delta_i$, it follows that,
    \begin{equation}
    (1 - \e) \fopt \leq \gamma_{i-1} \esp{f(S_{i-1})} + \gammaoff \esp{f(S_{i-1}')} \leq \gamma_{i-1} \delta_i \esp{f(S_{i})} + \gammaoff \esp{f(S'_{i})}\label{eq:bound1}
    \end{equation}
    where in the last inequality we have used the definition of $\delta_i$ and the fact that $f(S'_i) \geq f(S'_{i-1})$, which follows from the way $S'_i$ is defined in Algorithm~\ref{alg:multi-pass-random}.

    On the other hand, $\esp{f(S_{i})} - \esp{f(S_{i-1})} = ( 1-\delta_i)\esp{f(S_{i})}$. Thus, by \eqref{eq:non-monotone-main} we also have:
    \begin{align}
    (1- &{\e}) \fopt  \nonumber \\
    & \leq \lb \tfrac{p}{\beta_i} + p - 1 \!\rb \lb \esp{ f(S_i) } - \esp{f ( S_{i-1} )} \rb\!  +\! (p + \beta p + 1) \esp{f(S_{i})} +  \gammaoff \esp{f(S'_{i})}\nonumber\\
    & = \lb \lb \tfrac{p}{\beta_i} + p - 1 \rb \lb 1 - \delta_i \rb + p + \beta_i p + 1 \rb \esp{f(S_{i})} + \gammaoff \esp{f(S'_{i})}\,.\label{eq:bound2}
    \end{align}
    Since the right-hand side of equation~\ref{eq:bound1} is an increasing function of $\delta_i$ and the right-hand side of equation~\ref{eq:bound2} is a decreasing function of $\delta_i$, the guarantee we obtain is always at least as good as that obtained when these two values are equal.
    \end{proof}

    As in the monotone case, the lemma enables us to derive values of $\beta$ so as to minimize the value of the approximation ratio. The following follows directly from the same calculations as in Section~\ref{sec:select-param-beta_i} and Appendix~\ref{sec:calc-proof-theor}.
    \begin{theorem}\label{thm:non-monotone-approximation}
    Suppose we run Algorithm~\ref{alg:multi-pass-random} with a buffer of size $m = 4 d k /\e^2$ on a arbitrary $p$-matchoid constraint and a  submodular function, with $\alpha = \e f(\opt) /2k$, $\beta_1 = 1$ and $\beta_i = \tfrac{\gamma_{i-1} - 1 -p}{\gamma_{i-1}-1+p}$ where $\gamma_{i}$ is given by the recurrence, $\gamma_1 = 4p$ and $\gamma_i = \tfrac{4p \gamma_{i-1} (\gamma_{i-1} - 1 )}{(\gamma_{i-1} -1 + p)^2}$. Then,
    \begin{align*}
    ( 1- \e) \fopt \leq \lb p+1 + \frac{4p}{i} \rb \esp{f(\tilde{S}_i)} + {\gammaoff} \esp{f(S'_i)}.
    \end{align*}
    In particular after $d = \tfrac{4p}{\e}$ passes,
    \begin{align*}
        (1-\e) \fopt \leq \lb p + 1 + {\gammaoff} + \e \rb \esp{f(\bar{S}_d)}\,.
    \end{align*}
    Under a matroid constraint, Algorithm ~\ref{alg:multi-pass-random} with $\alpha = \e f(\opt) /2k$, $\beta_i = 1/i$ and $d = 2\e^{-1}$ passes outputs a solution $\bar{S}$ such that,
    \begin{align*}
    ( 1- \e) \fopt \leq \lb  2 + {\gammaoff} + \e \rb \esp{f(\bar{S})}\,,
    \end{align*}
    where $\gammaoff$ is the approximation ration of the best offline algorithm for maximizing $f$ under a matroid constraint.
    \end{theorem}


    \begin{proof}[Proof of Theorem \ref{thm:non-mono-matroid}]
      \label{proof:non-mono-matroid}
        We assume that we know the value of $f(\opt)$ before hand, which can be accomplished approximately as in Section~\ref{sec:guessing-value-opt}. Let $\e' = \e/p$ with $1/2 \geq \e' > 0$ and let $\alpha = \e' f(\opt)/2k$. We want to obtain an additive error term instead of a multiplicative error term as stated in Theorem~\ref{thm:non-monotone-approximation}. By  Theorem~\ref{thm:non-monotone-approximation},
    \begin{align*}
    ( 1- \e') \fopt & \leq \lb p+1 + \gammaoff + \frac{4p}{d} \rb \esp{f(\bar{S}_d)} \\
    & = \lb p + 1 + \gammaoff \rb \lb 1 + \bigO{d^{-1}} \rb \esp{f(\bar{S}_d)}\,.
    \end{align*}
    Using the fact that $(1 - \e')^{-1} \leq 1 + 2\e'$ for $\e' \in (0, 1/2]$, we get that,
    \begin{align}\label{eq:multiplicative-error}
    \fopt & \leq \lb p + 1 + \gammaoff \rb \lb 1 + \bigO{d^{-1}} \rb \lb 1 + 2 \e' \rb \esp{f(\bar{S}_d)}\,.
    \end{align}
    Since $\e' = \e/p$, setting $d = \bigO{p/\e}$ we finally obtain the desired result:
    \begin{align*}
    \fopt & \leq \lb p + 1 + \gammaoff \rb \lb 1 + \bigO{\e/p} \rb \lb 1 + 2 \e/p \rb \esp{f(\bar{S}_d)} \\
    & \leq \lb p + 1 + \gammaoff + \bigO{\e} \rb \esp{f(\bar{S}_d)}.
    \end{align*}

    For the space complexity, we note that the randomized local-search algorithm stores the buffer $B$ and maintains two past solutions $S_{i}, S'_{i} \in \cI$, together with the current solution $S \in \cI$. Hence, the total space needed is equal to $\bigO{\card{B} + \card{S'_i} + \card{S_i} + \card{S}} = \bigO{m + 3k} = \bigO{p^3k\e^{-3}}$, times an additional factor of $O(\log k)$ for guessing $f(\opt)$. The number of passes is $d = \bigO{p/\e}$.
  \end{proof}

\end{document}